\documentclass[a4paper,11pt,leqno]{article}
\usepackage{color}
\usepackage{amsmath}
\catcode`\@=11 \@addtoreset{equation}{section}

\catcode`\@=12

\usepackage{amssymb}
\usepackage{amsfonts}
\usepackage{xcolor}
\usepackage{graphics}
\usepackage{epsfig,psfrag,graphicx}
\usepackage{amssymb}
\usepackage{eepic,epic}
\usepackage{dsfont}

\usepackage{enumerate}

\usepackage{epsfig} % postscript figures
\usepackage{amsmath} % ams math
\usepackage{amsthm} %ams theorem
\usepackage{amssymb} % ams symbols
\usepackage{amsbsy} % ams bold symbols
\usepackage{bbm}
\usepackage{verbatim}

\theoremstyle{definition}
\newtheorem{definition}{Definition}[section]
\newtheorem{example}[definition]{Example}

\theoremstyle{plain}
\newtheorem{theorem}[definition]{Theorem}

\newtheorem{lemma}[definition]{Lemma}
\newtheorem{remark}[definition]{Remark}
\newcommand{\nc}{\newcommand}
\nc{\weak}{\rightharpoonup}
\nc{\weakstar}{\stackrel{\ast}{\rightharpoonup}}

\nc{\modular}[1]{{\stackrel{ #1}{\longrightarrow\,}}}

\allowdisplaybreaks

\title{\LARGE \bf{New characterizations of completely useful topologies in mathematical utility theory}}
\author{ \textsl{Gianni Bosi}$\,^{1}$, \textsl{Roberto Daris}$\,^{1}$ and \textsl{Gabriele Sbaiz}$\,^{1,}$\footnote{Corresponding author}
 \vspace{.2cm} \\
\footnotesize{$\,^1\;$ \textit{Department of Economics, Business, Mathematics and Statistics}} \\ \footnotesize{\textsc{University of Trieste}} \\
Via\footnotesize{ Valerio 4/1, 34127 Trieste, Italy} \vspace{0.1cm} \\
\vspace{.2cm} \\
\footnotesize{\ttfamily{gianni.bosi@deams.units.it}$\,,\quad$ \ttfamily{roberto.daris@deams.units.it}$\,,\quad$
\ttfamily{gabriele.sbaiz@deams.units.it}} \vspace{.1cm}
}

\date{\small March 1, 2024}

\begin{document}
\maketitle
%%%%%%%%%%%%%%%%%%%%%%%%%

\abstract{Let $X$ be an arbitrary set. Then a topology $t$ on $X$ is said to be \textit{completely useful} if every upper semicontinuous linear (total) preorder $\precsim$ on $X$ can be represented by an upper semicontinuous real-valued order preserving function. In this paper, appealing, simple and new characterizations of completely useful topologies will be proved, therefore clarifying the structure of such topologies.}



\section{Introduction}
Let $X$ be an arbitrary set. A {\em linear} ({\em total}) {\em preorder} $\precsim$ on $X$ is a {\em reflexive}, {\em transitive} and {\em linear} ({\em total}) binary relation on $X$. In Bosi and Herden \cite{B-H}, a topology $t$ on $X$ is said to be \textit{completely useful }if every {\em upper semicontinuous} linear (total) preorder $\precsim$ on $X$ can be represented by an {\em upper semicontinuous} real-valued {\em order preserving} function $f$ on the {\em topological linearly preordered space} $(X, t, \precsim)$.

We recall that a linear preorder $\precsim$ on $X$ is said to be upper semicontinuous if, for every point $x\in X$, the set $L(x):=\{y\in X: y\prec x\}$ is an open subset of $X$ (clearly, $\prec$ is the {\em strict} ({\em asymmetric}) part of $\precsim$). Further, a function $f:(X,\precsim ) \rightarrow ({\mathbb R}, \leq )$ is  said to be order preserving if $x \precsim y$ is equivalent to $f(x) \leq f(y)$ for all $x,y \in X$. In correspondence to the definition of an upper semicontinuous linear preorder $\precsim$ on $X$, the function $f: (X,t) \rightarrow ({\mathbb R}, t_{nat})$ is
said to be upper semicontinuous if $\{f<r\}:=\{z\in X:f(z)<r\}$ is an open subset of $X$ for every real number $r$ (here, $t_{nat}$ is the {\em natural topology} on $X$). Since perfectly analogous considerations hold for {\em lower semicontinuous linear preorders} and {\em lower semicontinuous real-valued order preserving functions}, respectively, we shall simply refer to semicontinuous linear (total) preorders and semicontinuous real-valued order preserving functions.

In Bosi and Herden \cite{B-H} the particular relevance of semicontinuous linear preorders and their representability by a semicontinuous real-valued order preserving function for the study of the interrelations between order and topology, and the foundations of mathematical utility theory has been motivated and underlined by many examples. Furthermore, in the same paper the fundamental problem of characterizing all completely useful topologies $t$ on $X$ has been solved to some satisfactory degree. Thus, the reader is referred to Bosi and Herden \cite{B-H} for a first discussion of this characterization problem. It is known that a completely useful topology $t$ on $X$ is in particular {\em useful}, in the sense that every {\em continuous linear preorder} admits a {\em continuous} order preserving function. It is worth noticing that the concept of a useful topology, which was introduced by Herden \cite{Herd3}, has received some attention in the literature in recent years (see, e.g., Bosi and Herden \cite{HBusefulstructure} and Bosi and Zuanon \cite{{Jota2020},Axioms}). On the contrary, the notion of a completely useful topology was not studied, despite for  a characterization by second countability of lower preorderable topologies, which was presented in Theorem 5.1 by Campi\'on et al. \cite{Cam}, and  the discussion, presented by Bosi and Franzoi \cite{Charup}, about the possibility of generalizing such concept to the case of non-total preorders.

A fundamental result (see Lemma 4.1 in \cite{B-H}) states that a completely useful topology $t$ on $X$ must be \textit{short}, i.e. there cannot exist any
uncountable ordinal $\alpha$ that can be order embedded into any of the ordered sets $\left( t,\subsetneqq\right) $ or $\left( t,\supsetneqq \right)$. The shortness of a topology $t$ on $X$ is equivalent to two fundamental properties of $t$. Indeed, $t$ is short if and only if $t$ is a \textit{hereditarily Lindel\"of }topology, i.e. for every subset $A$ of $X$ and every open covering $\mathcal{C}$ of $A$ there exists some countable subcovering $\mathcal{C}^{\prime}\subset\mathcal{C}$ of $A$, and \textit{hereditarily separable}, i.e. every subspace $\left( A,t_{|A}\right)$ of $\left( X,t\right)$ of some subset $A$ of $X$ is separable (see Proposition 4.2 in \cite{B-H}).

Unfortunately, shortness does not characterize completely useful topologies $t$ on $X$. Indeed, there exist (at least) two different types of counterexamples that show which a short topology $t$ on $X$ is not necessarily completely useful.

\begin{example}\label{ex1}
Let $X:=\mathbb{R}\times\{0,1\}$. Then we consider the natural lexicographic order $\leq_{L}$ on $X$ and choose the topology $t$ on $X$ that is generated by the sets
\begin{equation*}
d(r,i):=\{(s,j)\in X:(s,j)\leq_{L}(r,i)\}
\end{equation*}
for $r\in [0,1]$ and $i\in\{0,1\}$ and
$$K(t,k):=\{(s,j)\in X:(t,k)<_{L}(s,j)\}$$
for $t\in \mathbb{R}$ and $k=1$.

The definition of $t$ implies that $t$ is short. On the other hand, $t$ is not completely useful. Let, therefore, $\leq_{L}$ be the natural lexicographic order on $X$. Then the definition of $t$ implies that $\leq_{L}$ is a semicontinuous linear order on $X$. Since $\leq_{L}$ has uncountably many jumps it cannot be representable by a semicontinuous real-valued order preserving function $f$ on $X$. We refer to Section 2 in \cite{Jota2020} for the definition of jumps.

The chain $(\{L(x)\}_{x\in X},\subset)$ of open subsets of $X$ induces $\leq_{L}$ by setting
$$
x\leq_{L}y\Longleftrightarrow\forall z\in X, \, (y\in L(z)\Longrightarrow
x\in L(z))\, .
$$
The jumps of $\leq_{L}$ are defined by the uncountable family of open sets $\displaystyle L(s,0)=\bigcup_{(r,i)<_{L}(s,1)}L(r,i)\subsetneqq
L(s,1)\subsetneqq \bigcap_{(s,1)<_{L}(k,i)}L(k,i)=d(s,1)$, where $s$ runs through $[0,1]$, that have the property that their topological closures $\overline{L(s,0)}=\displaystyle \overline{\bigcup_{(r,i)<_{L}(s,1)}L(r,i)},\, \overline{L(s,1)}$ and $\displaystyle \overline{\bigcap_{(s,1)<_{L}(k,i)}L(k,i)}=\overline{d(s,1)}$ coincide. This property characterizes the first type of a short topology that is not completely useful.
\end{example}

\begin{example}\label{ex2}
A chain $(Z,\leq)$ is said to be \textit{short} if there exists no uncountable ordinal $\alpha$ that can be order-embedded into $(Z,<)$ or $(Z,>)$. An uncountable short chain $(A,\leq)$ is said to be an \textit{Aronszajn chain} if every subchain $(S,\leq)$ of $(A,\leq)$ that is representable by a real-valued order preserving function must be countable. Aronszajn chains have been discussed by Beardon et al. \cite{B-C-H} in connection with the general utility representation problem. The properties of Aronszajn chains or, equivalently, \textit{Specker types} are discussed in Baumgartner \cite{Bau}. Let now $(A,\leq)$ be an arbitrarily chosen Aronszajn chain that only has
countably many jumps.
%In ZFC (Zermelo-Fraenkel + Axiom of Choice) the existence of $(A,\leq)$ cannot be excluded (cf. Remark 3.3).
Then we consider the upper order topology $t_{u}^{\leq}$ on $A$ that is generated by the sets $L(a)$, where $a$ runs through $A$. The shortness of $\leq$ implies that $t_{u}^{\leq}$ is a short topology on $A$. Clearly, $\leq$ is a semicontinuous linear (total) preorder on $X$ that can be defined with help of the sets $L(a)$ in the same way as $\leq_{L}$ in Example \ref{ex1}. Since $(A,\leq)$ only has countably many jumps it follows that $t_{u}^{\leq}$ cannot contain uncountably many (open) sets $L(x)$ such that $\displaystyle \bigcup_{L(y)\subsetneqq L(x)}L(y)\subsetneqq L(x)\subsetneqq\bigcap\limits_{L(x)\subsetneqq L(z)}L(z)$ and $\displaystyle \overline{\bigcup_{L(y)\subsetneqq L(x)}L(y)}=\overline{L(x)}=\overline{\bigcap_{L(x)\subsetneqq L(z)}L(z)}$. This means that $t_{u}^{\leq}$ is of a different type than the topology that has been considered in Example \ref{ex1}. Moreover, the definition of an Aronszajn chain implies that $\leq$ cannot be represented by some semicontinuous real-valued order preserving function. In this counterexample the topological closure $\overline{O}$ of any non-empty open subset $O$ of $A$ is \textit{too thick}. Indeed, it coincides with $A$. This means that there exist open subsets $O$ of $A$ for which there exists some Aronszajn chain that can be order-embedded into the partially ordered set $(t_{\left(O,\overline{O}\right) },\subset)$ of all open subsets $O\subset O^{\prime}\subset \overline{O}$ of $A$. This property characterizes the second type of a short
topology that is not completely useful. The reader may notice that the topology that has been considered in Example \ref{ex1} is not too thick.
\end{example}

Because of these considerations, the problem arises if the afore-presented counterexamples are typical for short topologies $t$ on $X$ that are not completely useful. We, thus, conjecture, that a topology $t$ on $X$ is completely useful if and only if $t$ is short and the afore-presented types of counterexamples cannot hold. In the main Theorem \ref{thm:char} (see below) of this paper, we shall show that this conjecture is actually true. This means that the results in Bosi and Herden \cite{B-H} on completely useful topologies $t$ on $X$ can be completed by, in some sense, the best possible characterizations of completely useful topologies $t$ on $X$.

\medskip

The paper is structured as follows. Section \ref{s:preliminary} contains the definitions and some lemmas showing the interrelations among the axioms. Section \ref{s:mainthm} is devoted to the main theorem, containing the new characterizations of completely useful topologies. Section \ref{s:conclusions} concludes the paper and presents future directions of research.

\section{Preliminaries}\label{s:preliminary}
We introduce here the main notation and definitions which will be used throughout the whole paper. First, given the topological space $(X,t)$, denote  by $\mathcal{O}$  the collection of all {\em subchains} $(\mathbf{O,\subset)}$ of $(t,\subset)$.

\begin{definition}
Let $X$ be some arbitrarily but fixed chosen set and let $t$ be an arbitrary topology on $X$. 
\begin{itemize}

\item[(a)] For every set $O\in t$, we denote by $\mathcal{O}(O,\overline{O})$ the set of all chains $(\mathbf{O},\subset)\in \mathcal{O}$ such that $O\subset O^{\prime}\subset\overline{O}$ for every set $O^{\prime}\in \mathbf{O}$.

\item[(b)] For every chain ($\mathbf{O,\subset)\in}$ $\mathcal{O}$, we denote by $P(\mathbf{O)}$ the set of all pairs $O^{\prime}\subsetneqq O\in\mathbf{O}$
for which there exists some set $O^{\prime\prime}\in \mathbf{O}$ such that $O^{\prime}\subsetneqq O^{\prime\prime}\subsetneqq O$.

\item[(c)] $t$ is said to be \textit{strongly separable} if for every chain  $(\mathbf{O},\subset)\in \mathcal{O}$ there exists some countable subset $Y$ of $X$ such that $\left(Y\cap O\right)\setminus O^{\prime} \neq\emptyset$ for every pair of sets $O^{\prime}\subsetneqq O\in P(\mathbf{O})$.

\item[(d)] $t$ is said to be \textit{strongly thin} if there exists no chain $(\mathbf{O},\subset)\in \mathcal{O}$ that contains uncountably many sets $O$
such that $\displaystyle \bigcup_{\mathbf{O}\ni O^{\prime}\subsetneqq O}O^{\prime}\subsetneqq O\subsetneqq\bigcap_{O\subsetneqq O^{\prime\prime}\in\mathbf{O}} O^{\prime\prime}$ and $\displaystyle \overline{\bigcup_{\mathbf{O}\ni O^{\prime}\subsetneqq O}O^{\prime}}=\overline{O}=\overline{\bigcap_{O\subsetneqq O^{\prime
\prime}\in\mathbf{O}}O^{\prime\prime}}$.

\item[(e)] $t$ is said to be \textit{locally thin} if there exists no open subset $O$ of $X$ for which there exists some chain $(\mathbf{O},\subset)\in  \mathcal{O}(O,\overline{O})$ that is order-isomorphic to some Aronszajn chain.

\item[(f)] A collection $\left\{x_{i}\right\} _{i\in I}$ of points $x_{i}\in X$ is said to be \textit{weakly isolated} if the cardinality of $I$ is not greater than the cardinality of the real line and if there exists some chain $(\mathbf{O},\subset)\in \mathcal{O}$ and some function $\varphi:\left\{
x_{i}\right\}_{i\in I}\rightarrow \mathbf{O}$ such that $\displaystyle \bigcup_{\mathbf{O}\ni O^{\prime}\subsetneqq \varphi(x_{i})}O^{\prime}\subsetneqq \varphi(x_{i}) \subsetneqq \bigcap_{\varphi(x_{i})\subsetneqq O^{\prime\prime}\in\mathbf{O}}O^{\prime\prime}$, $\displaystyle \overline{\bigcup_{\mathbf{O} \ni O^{\prime} \subsetneqq \varphi(x_{i})}O^{\prime}}=\overline{\varphi(x_{i})}=\overline{\bigcap_{\varphi
(x_{i})\subsetneqq O^{\prime\prime}\in\mathbf{O}}O^{\prime\prime}}$, $\overline{\varphi(x_{i})}\neq \overline{\varphi(x_{j})}$, if $i\neq j$ and $\displaystyle x_{i}\in\varphi(x_{i}) \setminus \bigcup_{\mathbf{O}\ni O^{\prime}\subsetneqq \varphi(x_{i})} O^{\prime}$.
\end{itemize}
\end{definition}

\begin{remark}
Our main Theorem \ref{thm:char} (see below) implies, in particular, that a strongly separable topology $t$ on $X$ is separable. The converse does not hold (see Examples \ref{ex1} and \ref{ex2}). This consideration justifies the concept of a strongly separable topology $t$ on $X$ given in the previous definition.
\end{remark}

Let now $(\mathbf{O,\subset)\in}$ $\mathcal{O}$ be arbitrarily chosen. Then a \textit{gap} of $\mathbf{O}$ is a pair $(O,B)$ of subsets of $X$ that satisfies the following conditions:
\begin{itemize}
\item[\textbf{G1:}] $\displaystyle O=\bigcup_{\mathbf{O}\ni O^{\prime}\subset O}O^{\prime }$ and $\displaystyle B=\bigcap_{O\subsetneqq O^{\prime \prime}\in \mathbf{O}}O^{\prime \prime}$.

\item[\textbf{G2:}] There exists an open set $O^{+}\in t\setminus\mathbf{O}$ such that $O\subsetneqq O^{+}\subsetneqq B$.
\end{itemize}
We say that $(\mathbf{O},\subset)\in \mathcal{O}$ is \textit{gap free} if it has no gaps.

\medskip

At this point, we give the definition of a {\em thin} topology.

\begin{definition}
A topology $t$ on $X$ is said to be \textit{thin} if it satisfies the
following conditions:
\begin{itemize}
\item[\textbf{T1:}] For every weakly isolated collection $\left\{ x_{i}\right\}_{i\in I}$ of points $x_{i}\in X$, there exists some open covering $\mathcal{C}$ of $\left\{ x_{i}\right\} _{i\in I}$ such that every set $O\in\mathcal{C}$ contains at most countably many points $x_{i}$.

\item[\textbf{T2:}] There exists no open subset $O$ of $X$ for which there exists some chain $(\mathbf{O},\subset)\in\mathcal{O}(O,\overline{O})$ the
cardinality of which is not greater than the cardinality of the real line and that has uncountably many gaps.
\end{itemize}
\end{definition}

\begin{remark}
It is worth noting that, the topology $t$ on $X$ that has been considered in Example \ref{ex2} is thin but not locally thin. On the other hand, the topology $t$ on $X$ that has been considered in Example \ref{ex1} is locally thin but not thin.
\end{remark}

Now the following lemma holds.

\begin{lemma}\label{lem:thin}
Let $t$ be a short topology on $X$. Then in order for $t$ to be thin it is necessary and sufficient that $t$ is strongly thin.
\end{lemma}
\begin{proof}
The sufficiency part of the lemma is trivial. The proof of the necessity part is based upon the results in Herden and Pallack \cite{H-P} that, in
particular, imply that a short topology $t$ on $X$ cannot contain chains $\left( O,\subset\right) \in\mathcal{O}$ the cardinality of which is greater
than the cardinality of the real line. Applying this result and our assumption that $t$ is a hereditarily Lindel\"of topology on $X$, condition \textbf{T1} implies, with help of the transfinite induction argument as it has been used, for instance, in the first part of the proof of Proposition 4.2 in Bosi and Herden \cite{B-H}, that every weakly isolated collection $\left\{x_{i}\right\} _{i\in I}$ of points $x_{i}\in X$ must be countable. In addition, we may conclude from condition \textbf{T2} with help of this result that, for every open subset $O$ of $X$, each chain $\left( \mathbf{O},\subset\right) \in\mathcal{O}\left( O,\overline{O}\right) $ has at most countably many gaps. Summarizing these conclusions it follows
that $t$ must be strongly thin.
\end{proof}

Let us now consider an open subset $O$ of $X$ and some chain $(\mathbf{O},\subset)\in \mathcal{O}(O,\overline{O})$. Then the definition of a strongly thin topology $t$ on $X$ immediately implies the following lemma, whose proof is omitted.

\begin{lemma}\label{lem:2.2}
Let $t$ be a strongly thin topology on $X$. Then $(\mathbf{O,\subset)}$ has at most countably many gaps.
\end{lemma}

Now, we recall that $t$ is said to satisfy \textbf{TIP} (\textbf{T}ransfinite \textbf{I}nduction \textbf{P}rocedure) if, for every open subset $O$ of $X$ and every pair of open subsets $O\subset O_{0}^{\prime}\subsetneqq O_{0}\subset\overline{O}$ of $X$, the following construction by transfinite induction always leads to a countable subchain of $\mathcal{O}(O,\overline {O})$:
\begin{itemize}
\item In the first step we set $\mathbf{O}_{0}:=\{O_{0}^{^{\prime}},O_{0}\}$.

\item At non-limit steps $\alpha$ we choose an arbitrary gap $(O,B)$ of $\mathbf{O}_{\alpha-1}$ and some open set $O^{+}\in t\setminus\mathbf{O}_{\alpha-1}$ such that $O\subsetneqq O^{+}\subsetneqq B$. Then $\mathbf{O}_{\alpha}$ is the union of $\mathbf{O}_{\alpha-1}$ with $\left\{ O^{+}\right\}$. In case that $\mathbf{O}_{\alpha-1}$ has no gaps we set $\mathbf{O:=O}_{\alpha-1}\,$ and this finishes the transfinite induction process.

\item At limit steps $\alpha$ we set $\mathbf{O}_{\alpha}:=\bigcup\limits_{\beta<\alpha}\mathbf{O}_{\beta}$.
\end{itemize}

Let now $\leq$ be an arbitrary semicontinuous linear (total) preorder on $X$ and consider for every point $x\in X$ the corresponding open subset $L(x)$. Finally, Lemma \ref{lem:2.2} implies, with help of a transfinite induction argument (that is similar to the construction in the appendix of Beardon et al. \cite{B-C-H}), the validity of the following lemma.

\begin{lemma}
Let $t$ be a strongly thin topology on $X$. Then in order for $t$ to be locally thin it is necessary
and sufficient that $t$ satisfies \textbf{TIP}.
\end{lemma}

\section{The characterization theorem}\label{s:mainthm}

\begin{theorem}\label{thm:char}
Let $t$ be an arbitrary topology on $X$. Then the following assertions are equivalent:
\begin{itemize}
\item[(i)] $t$ is completely useful.

\item[(ii)] $t$ is strongly separable.

\item[(iii)] $t$ is short, thin and locally thin.

\item[(iv)] $t$ is a hereditarily separable, thin, locally thin and hereditarily Lindel\"of topology.
\end{itemize}
\end{theorem}

\begin{proof}
$(i) \Longleftrightarrow (ii):$ This equivalence corresponds to the equivalence of the assertions $(i)$ and $(iv)$ of Theorem 4.11 in Bosi and Herden \cite{B-H}.

$(i)\wedge (ii)\Longrightarrow (iii):$ The implication ``$(v) \Longrightarrow (vi)$'' of Theorem 4.11 or,
alternatively, Lemma 4.1 in Bosi and Herden \cite{B-H} imply that $t$ is short. In addition, the afore-presented transfinite induction procedure, in combination with the arguments of the proof of the implication ``$(i) \Longrightarrow (ii)$'' of Theorem 4.11 in Bosi and Herden \cite{B-H}, allow us to conclude that $t$ must be locally thin. Finally, it is an easy task to verify that a strongly separable topology $t$ on $X$ is strongly thin and, thus, also thin.

$(iii)\Longleftrightarrow (iv):$ This equivalence is an immediate consequence of Proposition 4.2 in Bosi and Herden \cite{B-H}.

$(iii) \wedge (iv) \Longrightarrow (ii):$ Let $(\mathbf{O},\subset)\in \mathcal{O}$ be an arbitrarily chosen chain. Then two sets $O^{\prime},O\in\mathbf{O}$ are said to be equivalent if $\overline{O^{\prime}}=\overline{O}$. The corresponding equivalence classes are abbreviated as usual by $\left[ O\right]$. Now we choose in every equivalence class $\left[ O\right]$ some fixed set $O$. The subchain of $(\mathbf{O},\subset)$ that consists of these sets $O$ is denoted by $(\mathbf{O}^{\prime},\subset)$. In a first step we show that there exists some countable subset $Y^{\prime}$ of $X$ such that $Y^{\prime}\cap(O\setminus O^{\prime})\neq\emptyset$ for every pair of sets $O^{\prime}\subsetneqq O\in\mathbf{O}^{\prime}$. Therefore, we have to distinguish the following four cases: 
\begin{itemize}
\item[(1)] $(\mathbf{O}^{\prime},\subset)$ have a first and a last element; 
\item[(2)] $(\mathbf{O}^{\prime},\subset)$ have a first but no last element; 
\item[(3)] $(\mathbf{O}^{\prime},\subset)$ have no first but a last element; 
\item[(4)] $(\mathbf{O}^{\prime},\subset)$ have neither a first nor a last element.
\end{itemize} 
Since all these cases can be settled by analogous arguments we concentrate on the case $(\mathbf{O}^{\prime},\subset)$ to neither have a first nor a last element. Let $(\mathbb{Z},\leq)$ be the chain of integers. The shortness of $t$ implies the existence
of some countable subchain $(\mathbf{O}_{0}^{\prime},\subset):=(\{O_{z}\}_{z\in \mathbb{Z}},\subset)$ of $(\mathbf{O}^{\prime},\subset)$ such that for every set $O\in\mathbf{O}^{\prime}$ there exist sets $O_{z},O_{z^{\prime}}\in \mathbf{O}_{0}^{\prime}$ such that $O_{z}\subset O\subset O_{z^{\prime
}}$ and $O_{z}\subset O_{z^{\prime}}\Longleftrightarrow z\leq z^{\prime}$. Now we set $U_{z}:=O_{z}\setminus\overline{O}_{z-1}$ for
every $z\in\mathbb{Z}$. The definition of $\mathbf{O}^{\prime}$ implies that $U_{z}\neq\emptyset$ for every $z\in\mathbb{Z}$ and that $U_{z}\cap U_{z^{\prime}}=\emptyset$ for every pair of different integers $z,z^{\prime}$. In this way we, thus, have obtained a
countable set $T_{0}:=\{U_{z}:z\in\mathbb{Z}\}$ of pairwise disjoint open subsets of $X$. Starting with $\mathbf{O}_{0}^{\prime}$ and $T_{0}$ we now construct by transfinite induction a tree $(T,\supset)$ that at each level $\gamma$ consists of pairwise disjoint (non-empty) open subsets of $X$.

Let, therefore, $0<\alpha$ be not a limit ordinal. Then we consider the set $G(\mathbf{O}_{\alpha-1}^{\prime})$ of all pairs $(O,B)$ that satisfy with
respect to $\mathbf{O}_{\alpha-1}^{\prime}$ condition \textbf{G1} of the definition of a gap. In case that $G(\mathbf{O}_{\alpha-1}^{\prime})=\emptyset$ we set $(T,\supset):=(T_{\alpha-1},\supset)$ and this finishes the transfinite induction process. Otherwise, we choose, for every pair $(O,B)\in G(\mathbf{O}_{\alpha-1}^{\prime})$, an arbitrary pair of sets $O^{\prime\prime}\subsetneqq O^{\prime}\in \mathbf{O}^{\prime}$ such that $O\subset O^{\prime\prime}\subsetneqq O^{\prime}\subset B$ in order to then consider in this way the obtained
additional open sets $O^{\prime}\setminus\overline{O^{\prime\prime}}$. Because of the definition of $\mathbf{O}^{\prime}$ no additional set $O^{\prime}\setminus\overline{O^{\prime\prime}}$ is empty. This conclusion allows us to define $\mathbf{O}_{\alpha}^{\prime}$ as the union of $\mathbf{O}_{\alpha-1}^{\prime}$ with all new sets $O^{\prime\prime}$ and $O^{\prime}$ and $T_{\alpha}$ as the union of $T_{\alpha-1}$ with all additional open subsets $O^{\prime}\setminus \overline{O^{\prime\prime}}$of $X$ that have been obtained in the afore-described way. Obviously, $(T_{\alpha},\supset)$ is a tree that at each level $\gamma$ consists of pairwise disjoint (non-empty) open subsets of $X$.

In case that $\alpha$ is a limit ordinal we set $\displaystyle \mathbf{O}_{\alpha}^{\prime}:=\bigcup_{\beta<\alpha}\mathbf{O}_{\beta}^{\prime}$ and $\displaystyle T_{\alpha}:=\bigcup_{\beta<\alpha}T_{\beta} $. Also in this case $(T_{\alpha},\supset)$, clearly, is a tree that at each level $\gamma$ consists of pairwise disjoint (non-empty) open subsets of $X$.

Since $t$ is short $T$ cannot contain uncountably many pairwise disjoint open sets. Because of the construction of $(T,\supset)$ we have, in particular, that $(T,\supset)$ cannot contain branches of uncountable length and that $\left(T,\supset\right)$ at each level $\gamma$ only contains
countably many branches. The reader may notice that $\gamma$ does not necessarily correspond to the ordinal $\alpha$ that has been considered in the construction of $\left( T,\supset\right) $. Of course, the construction of $(T,\supset)$ does not exclude that $(T,\supset)$ is an \textit{Aronszajn tree} (see, for instance, Jech \cite{Jec}), because the previous observations only imply that the least upper bound of the lengths of all branches of $(T,\supset)$ is $\aleph_{1}$. We now show that the least upper bound of the lengths of all branches of $(T,\supset)$ cannot be greater than the first infinite cardinal $\aleph_{0}$. Indeed, the separability of $t$ implies the existence of some countable subset $S$ of $X$ such that $\overline{S}=X$. The countability of $S$ implies with help of the construction of $(T,\supset)$ that there exists some countable ordinal $\alpha$ such that $S\cap U=\emptyset$ for every ordinal $\xi >\alpha$ and every open set $U\in T_{\xi}\setminus T_{\alpha}$, which contradicts the property that $S$ is a dense subset of $X$. Hence, we may actually conclude that the least upper bound of the lengths of all branches of $(T,\supset)$ cannot be greater than $\aleph_{0}$. We recall that an analogous argument implies that a separable chain does not allow the construction of an Aronszajn tree $(T,\supset)$ that consists of non-empty open intervals that at each level of $(T,\supset)$ are pairwise disjoint. Since $(T,\supset)$ is not an Aronszajn tree, we may summarize our considerations in order to conclude that $T$ is a countable set. By choosing in every open set $U\in T$ some point $y$ we, thus, obtain a countable subset $Y^{\prime}$ of $X$. Let now some pair $O^{\prime}\subsetneqq O\in\mathbf{O}^{\prime}$ of sets be arbitrarily chosen. At this point, we have to prove that $(Y^{\prime}\cap O)\setminus O^{\prime}\neq\emptyset$. The construction of $T_{0}$ implies that there exists some set $U_{z}\in T_{0}$ such that $(U_{z}\cap O)\setminus O^{\prime}\neq\emptyset$. Hence, an analysis of the construction of $(T,\supset)$ allows us to conclude that there exists some ordinal $\alpha$ and some set $U\in T_{\alpha}$ such that $U\subset (O\setminus O^{\prime})$, which means that $(Y^{\prime}\cap O)\setminus O^{\prime}\neq\emptyset$.

Now, since $t$ is locally thin we may apply the chain ``$(ii)\Longrightarrow (iii)\Longrightarrow (v)$'' of implications of the proof of Theorem 4.11 in Herden and Bosi \cite{B-H}. It follows that for every equivalence class $\left[ O\right] $ there exists some countable subset $Y_{O}$ of $X$ such that $(Y_{O}\cap O^{\prime})\setminus O^{\prime\prime}\neq\emptyset$ for every pair of sets $O^{\prime\prime}\subsetneqq O^{\prime}\in P\left( \left[ O\right]\right)$. Therefore, the proof of the desired implication will be finished if we are able to show that there exist at most countably many equivalence classes $\left[ O\right]$ that contain pairs of sets $O^{\prime\prime}\subsetneqq O\subsetneqq O^{\prime}\in\mathbf{O}$. Let $\mathcal{E}$ be the set of these crucial equivalence classes. Indeed, in case that $\mathcal{E}$ is a countable set we may take $\displaystyle Y:=Y^{\prime}\cup \bigcup_{\left[ O\right] \in\mathcal{E}}Y_{O}$ and nothing remains to be shown. In order to prove that $\mathcal{E}$ is a countable set we choose in every equivalence class $\left[ O\right] \in\mathcal{E}$ arbitrary but fixed sets $O^{\prime\prime}\subsetneqq O\subsetneqq O^{\prime}\in\mathbf{O}$ and consider the subchain $(\mathbf{O}^{\prime\prime},\subset)$ of $(\mathbf{O},\subset)$ that consists of these sets $O^{\prime\prime}\subsetneqq O\subsetneqq O^{\prime}$. Since $t$ is thin and, due to Lemma \ref{lem:thin} also strongly thin, we may conclude that $(\mathbf{O}^{\prime\prime},\subset)$ is a countable chain, which implies that $\mathcal{E}$ actually is a countable set. This last conclusion completes the proof of the theorem.
\end{proof}

Throughout the literature, the only theorem on completely useful topologies that is well known is the Rader's theorem \cite{Rad}, which states that every second countable topology $t$ on $X$ is completely useful. Unfortunately Rader's proof of this theorem contained the same mistake as the first proof of Debreu \cite{Deb54} of his famous Open Gap Lemma. This mistake has been discovered by Mehta \cite{Meh}. Meanwhile, there exist several correct proofs of Rader's theorem (see, for instance, Isler \cite{Isl} or Richter \cite{Ric}). In this context, Theorem \ref{thm:char} widely generalizes Rader's theorem. Indeed, there may even exist completely useful Hausdorff-topologies $t$ on $X$ that are not first countable. Let, therefore, $X:=[0,1]$ where $[0,1]$ denotes the standard real interval. Then we consider the topology $t$ on $X$ that is generated by the closed, respectively half open half closed, intervals $[0,r]$ and $]s,1]$, where $r$ runs through all reals that are greater than $0$ but not greater than $1$ and $s$ runs through all reals that are smaller than $1$ but not smaller than $0$. Clearly, Theorem \ref{thm:char} implies that $t$ is a completely useful topology on $X$. Furthermore, the definition of $t$ implies that $t$ is a Hausdorff-topology on $X$ that is not first countable.

It seems that Theorem \ref{thm:char} hardly can be improved. The assumption that $t$ must be thin cannot be weakened. Indeed, there exist even compact Hausdorff-spaces that are short and locally thin but not thin as the following example shows.

\begin{example}
Let $[0,1]$ be the set of all reals that are not smaller than $0$ and not greater than $1$. Then we choose the set $X:=[0,1]\times\{0,1\}$ endowed with its natural (linear) lexicographic order $\leq_{L}$. In a similar way as in Example \ref{ex1} we consider the topology $t$ on $X$ that is generated by the sets
$$d(r,i):=\{(s,j)\in X:(s,j)\leq_{L}(r,i)\}$$
for $r\in [0,1]$ and $i\in\{0,1\}$,
$$K(t,k):=\{(s,j)\in X:(t,k)<_{L}(s,j)\}$$
for $t\in\lbrack0,1]$ and $k=1$,
$$U_{0}:=\{(r,0):r [0,1]\}$$
and
$$U_{1}:=\{(r,1): r\in [0,1]\}\, .$$

With help of the definition of $t$ and the standard argument that proves the compactness of the real interval $[0,1]$ it follows that $t_{\big|U_{0}}$ as
well as $t_{\big|U_{1}}$ are compact Hausdorff-topologies on $U_{0}$ and $U_{1}$ respectively. Since $U_{0}$ and $U_{1}$ define a partition of $X$ into two
disjoint open subsets of $X$ we, thus, may conclude that $\left( X,t\right)$ is a compact Hausdorff-space. In addition, we may conclude as in Example \ref{ex1} that $t$ is short, locally thin but not thin. This concludes the example.
\end{example}

\section{Conclusions} \label{s:conclusions}
In this paper we have presented several characterizations of completely useful topologies, i.e. topologies for which every upper semicontinuous linear (total) preorder admits an upper semicontinuous order preserving function. Such characterizations are based on the types of non-representability of chains. Results of this kind are interesting not only from a purely theoretical point of view, but also for their (well-known) applications to Economics and Decision Theory. In a future paper, we shall consider the possibility of incorporating the {\em Souslin Hypothesis} (see \cite{Sou}), in order to simplify the characterizations.

\paragraph*{Acknowledgements.}
G. Sbaiz is member of the INdAM (Italian Institute for Advanced Mathematics) group.

\end{document}